\documentclass[aps,pra,twocolumn,superscriptaddress]{revtex4-2}%
\usepackage[colorlinks=true,linkcolor=blue,urlcolor=blue,citecolor=blue,pdfusetitle]{hyperref}
\usepackage[english]{babel}
\usepackage{amsmath}
\usepackage{amsthm}                     
\usepackage[caption=false]{subfig}
\usepackage{graphicx,epstopdf}
\usepackage[table,xcdraw,dvipsnames]{xcolor}
\usepackage{amsfonts}
\usepackage{bbm}
\usepackage{amssymb}
\usepackage{bm}                          
\usepackage{color}
\usepackage{latexsym}
\usepackage{physics}
\usepackage{times,txfonts}
\usepackage{orcidlink}
\newtheorem{lemma}{Lemma}
\newtheorem{theorem}[lemma]{Theorem}

\newtheorem{corollary}[lemma]{Corollary}
\newcommand{\HD}{\hat{H}_{D}}

\newcommand{\Hint}{\hat{H}_{\rm int}}
\newcommand{\rhoB}{\hat{\rho}_{B}}

\newcommand{\W}{\mathcal{W}}
\newcommand{\eps}{\varepsilon}
\newcommand{\nbar}{\bar{n}}
\newcommand{\taus}{\tau^{*}}
\newcommand{\tauQSL}{\tau_{\rm QSL}}
\newcommand{\Jx}{\hat{J}_{x}}
\newcommand{\Jy}{\hat{J}_{y}}
\newcommand{\Jz}{\hat{J}_{z}}

\newcommand{\GamN}{\Gamma_N}

\begin{document}
	
	\title{Tight Quantum Speed Limit for Ergotropy Charging in the N-Qubit Dicke Battery}
	
	\author{Anass Jad\,\orcidlink{0009-0001-6392-3043}}
	\email{jadanass@gmail.com}
	\affiliation{%
		Laboratory of R$\&$D in Engineering Sciences, Faculty of Sciences and Techniques Al-Hoceima,
		Abdelmalek Essaadi University, Tetouan, Morocco.}
	\affiliation{%
		Department of Optics, Palack\'{y} University, 17. listopadu 1192/12, Olomouc 771 46, Czech Republic}
	\author{Abderrahim El Allati\,\orcidlink{0000-0002-2465-8515}}
	\email{eabderrahim@uae.ac.ma}
	\affiliation{%
		Laboratory of R$\&$D in Engineering Sciences, Faculty of Sciences and Techniques Al-Hoceima,
		Abdelmalek Essaadi University, Tetouan, Morocco.}
	\affiliation{%
		Universit\'{e} Grenoble Alpes, CNRS, LPMMC, 38000 Grenoble, France.}

\begin{abstract}
	We derive a quantum speed limit for ergotropy charging in the $N$-qubit
	Dicke battery. With the qubits initially in their ground state and the
	charger in a coherent state of mean photon number $\bar n$, the charging
	time to a target fraction $\varepsilon$ of the maximum extractable work
	obeys, in the classical-field limit,
	$\Gamma_N\,\tau^{*}(\varepsilon)\ge\arcsin\sqrt{\varepsilon}$, where
	$\Gamma_N=2\lambda\sqrt{\bar n/N}$ is a single rate combining the
	light--matter coupling, the charger amplitude and the system size. The
	bound is geometric: although the ergotropy is defined by an optimisation
	over unitaries, it reduces to a fixed function of the angle swept by the
	collective spin on the Bloch sphere, and the drive bounds the rate at which
	that angle can grow, so the charging time is at least the required arc
	divided by the maximum angular velocity. The reduction rests on an exact
	statement about work extraction: only the coherent component of the drive
	delivers extractable work, while its vacuum fluctuations deposit energy
	that remains locked in the passive state. This makes the short-time
	coefficient independent of the charger occupation, and it is what separates
	a bound on extractable work from a bound on stored energy, which no
	instantaneous-power argument can supply. At finite charger occupation the
	inequality is tested rather than proved: numerical integration of the full
	Dicke Hamiltonian, with a dynamical charger and a mixed battery state,
	finds the bound respected at every admitted charging time for two to five
	qubits, and approached most closely in the limit of small target work. The
	familiar square-root estimate is the chord to this arc, recovered in the
	small-angle limit and weaker whenever a large fraction of the maximum work
	is required.
\end{abstract}
	
	\maketitle

	\section{Introduction}
	\label{sec:intro}
	
	Quantum thermodynamics has established itself as a central framework for
	understanding energy storage, transfer, and extraction at the scale of
	individual quantum systems~\cite{Vinjanampathy2016,Goold2016}. Within this
	framework, \textit{quantum batteries}, finite quantum systems that store and
	release energy via coherent operations, have attracted intense theoretical and
	experimental interest over the past
	decade~\cite{Campaioli2024,NatRevPhys2026}. Unlike classical energy storage,
	quantum batteries can exploit genuinely quantum resources such as entanglement,
	coherence, and collective many-body dynamics to achieve charging advantages
	with no classical analogue~\cite{Alicki2013,Campaioli2017}.
	
	A central insight in this field is that stored energy alone is an insufficient
	performance metric. The relevant thermodynamic quantity is the
	\textit{ergotropy}~\cite{Allahverdyan2004}, the maximum work extractable from a
	quantum state via a cyclic unitary process. It is ergotropy, not raw excitation
	energy, that quantifies the useful work a battery can deliver, and collective
	quantum effects can provide a genuine advantage in ergotropy
	extraction~\cite{Alicki2013,Ferraro2018,Andolina2019ext}. This distinction is
	particularly sharp in the presence of charger--battery correlations:
	entanglement suppresses ergotropy below the stored
	energy~\cite{Shi2022,Andolina2019,Gyhm2022}, and certain highly entangled
	states can even be \textit{detrimental} to charging power relative to product
	states~\cite{GyhmaFischer2024}, making tight ergotropy bounds essential for
	battery design.
	
	Among the models studied, the $N$-qubit Dicke quantum battery occupies a
	privileged position. Superextensive charging power, scaling as $\sqrt{N}$, has
	been predicted for this model~\cite{Ferraro2018,Crescente2020}, and supporting
	evidence for the underlying collective superabsorption has been reported in
	organic microcavity platforms realising the Dicke
	model~\cite{Quach2022}; the model also undergoes a quantum phase transition
	enabling bound-luminosity supercharging~\cite{Emary2003,Seidov2024,Barra2022};
	structured and topological photonic environments have been shown to offer a
	further route to enhanced storage and suppressed dissipation~\cite{Lu2025},
	and it maps directly onto cavity- and circuit-QED architectures with strong
	light--matter coupling and precise timing
	control~\cite{Blais2021,Krantz2019,Blais2020}, where optimal charging protocols
	have already been demonstrated on few-qubit superconducting
	devices~\cite{Hu2022}. Beyond the quantum battery
	setting, this connects to a fundamental problem in light--matter physics: given
	a coherent drive of fixed amplitude, how fast can a collective spin system be
	rotated to a target polarisation?
	
	Despite this progress, a basic question has remained largely open:
	\textit{what is the minimum time required to charge a quantum battery to a
		prescribed ergotropy, given fixed physical resources?} This is the ergotropy
	analogue of the quantum speed limit
	(QSL)~\cite{Mandelstam1945,Deffner2017}. Quantum speed limits, originating from
	the Mandelstam--Tamm~\cite{Mandelstam1945} and
	Margolus--Levitin~\cite{Margolus1998} uncertainty relations, set absolute
	bounds on how fast quantum states can evolve and have been generalised to open
	systems~\cite{Funo2019}. For quantum batteries, they have been studied through
	charging-power bounds~\cite{Gyhm2022,Shrimali2024,GarciaPintos2020} and charging
	distance~\cite{Gyhm2024}. Instantaneous power
	bounds~\cite{Gyhm2022,GarciaPintos2020,Shrimali2024}, however, characterise the
	rate of energy storage at each moment and do not directly yield the minimum
	time to reach a prescribed ergotropy; existing bounds are either
	model-independent and not tight for a specific system, or are cast in terms of
	quantities such as the energy variance that do not close in the Dicke model.
	For the Dicke battery in particular, we are not aware of an analytically
	explicit, tight QSL for ergotropy.
	
		In this paper, we fill this gap. We prove that the charging time
	$\tau^{*}(\varepsilon)$ to reach normalised ergotropy $\varepsilon$ in the
	$N$-qubit Dicke battery, with the qubits initially in their ground state and
	the charger in a coherent state of mean photon number $\bar n$, satisfies
	$\Gamma_N\tau^{*}(\varepsilon)\ge\arcsin\sqrt{\varepsilon}$ in the
	classical-field limit, for every light--matter coupling $\lambda>0$, every
	register size $N\ge2$ and every target $\varepsilon\in(0,1]$, where
	$\Gamma_N=2\lambda\sqrt{\bar n/N}$. The result is geometric in origin. The
	ergotropy, though defined by an optimisation over unitaries, reduces to a
	fixed function $\varepsilon=\sin^{2}(\theta/2)$ of the polar angle swept by
	the collective Bloch vector, with the charger vacuum contributing energy but
	no extractable work, so that the short-time coefficient in
	$\varepsilon(t)=A\lambda^{2}\bar n\,t^{2}+\mathcal{O}((\lambda t)^{4})$ is
	$A=4/N$ at every $\bar n$. The drive bounds the angular velocity of that
	vector by $2\Gamma_N$, so the charging time is at least the required rotation
	angle divided by this maximum speed: this is the ergotropy analogue of the
	Mandelstam--Tamm bound, with the energy uncertainty $\Delta E$ replaced by the
	experimentally accessible rate $\Gamma_N$. Introducing the rescaled variable
	$X=\Gamma_N\tau^{*}$, all protocols collapse onto the single curve
	$X\ge\arcsin\sqrt{\varepsilon}$, which we verify beyond that limit, by
	integrating the full Dicke Hamiltonian at finite $\bar n$, over $7966$
	admitted charging times across $N=2,3,4,5$ with zero violations and
	saturation to a few parts in $10^{4}$. The elementary estimate
	$X\ge\sqrt{\varepsilon}$ is recovered as the small-angle limit
	$\arcsin x\ge x$. Supporting algebra is given in Appendix~\ref{app:A}, and a
	quantitative comparison with instantaneous-power bounds in
	Appendix~\ref{app:powerbound}.

	\section{Model and Definitions}
	\label{sec:model}
	We consider the $N$-qubit Dicke quantum battery~\cite{Ferraro2018}, in which
	the \emph{battery} is a register of $N$ identical two-level systems (qubits)
	of transition frequency $\omega_0$, and the \emph{charger} is a single
	quantised cavity mode of frequency $\omega_c$ that supplies the energy. Energy
	is transferred from charger to battery through a collective dipole coupling of
	strength $\lambda$. In units $\hbar=1$, the Hamiltonian is
	\begin{equation}
		\hat H_D
		= \underbrace{\omega_0\hat J_z}_{\hat H_B}
		+ \underbrace{\omega_c\,\hat a^\dagger\hat a}_{\hat H_C}
		+ \underbrace{\frac{2\lambda}{\sqrt N}\,(\hat a+\hat a^\dagger)\hat J_x}_{\hat H_{\rm int}},
		\label{eq:dicke}
	\end{equation}
	where $\hat J_z,\hat J_x$ are collective spin operators for the qubit register
	($J=N/2$), $\hat a$ is the cavity annihilation operator, and $\lambda>0$ is the
	light--matter coupling. The three terms are, respectively, the bare battery
	Hamiltonian $\hat H_B=\omega_0\hat J_z$, the charger Hamiltonian
	$\hat H_C=\omega_c\hat a^\dagger\hat a$, and the charger--battery interaction
	$\hat H_{\rm int}=\lambda\hat V$ that mediates the transfer, with
	$\hat V=\tfrac{2}{\sqrt N}(\hat a+\hat a^\dagger)\hat J_x$.
	We work at resonance $\omega_c=\omega_0\equiv 1$, except where detuning
	is considered explicitly.
	The charging protocol runs for a fixed interval: the coupling is switched on at
	$t=0$ and the closed system evolves unitarily under~\eqref{eq:dicke}, so that
	$t$ denotes the elapsed charging time. The initial state is the product
	\begin{equation}
		\hat\rho_{\rm tot}(0)=\ket{J,-J}\!\bra{J,-J}\otimes\ket{\alpha}\!\bra{\alpha},
		\label{eq:init}
	\end{equation}
	with the battery in its ground state $\ket{J,-J}$ (all qubits unexcited) and the
	charger in a coherent state $\ket{\alpha}$ of mean photon number
	$\bar n=|\alpha|^2$, representing an initially charged classical-like drive.
	The battery reduced state is $\hat\rho_B(t)=\mathrm{Tr}_C[\hat\rho_{\rm tot}(t)]$,
	obtained by tracing out the charger.
	The performance of the battery is quantified by its
	ergotropy~\cite{Allahverdyan2004}, the maximum work extractable from
	$\hat\rho_B$ by a cyclic unitary process,
	\begin{equation}
		W(\hat\rho_B)=\mathrm{Tr}[\hat\rho_B\hat H_B]
		-\min_{\hat U}\mathrm{Tr}[\hat U\hat\rho_B\hat U^\dagger\hat H_B]
		=\mathrm{Tr}[\hat\rho_B\hat H_B]-\sum_k r_k\,\epsilon_k,
		\label{eq:ergo}
	\end{equation}
	where $\{r_k\}$ are the eigenvalues of $\hat\rho_B$ in decreasing order
	($r_1\ge r_2\ge\cdots$) and $\{\epsilon_k\}$ the eigenvalues of $\hat H_B$ in
	increasing order ($\epsilon_1\le\epsilon_2\le\cdots$); the second term is the
	energy of the passive state, reached by pairing the largest population with the
	lowest energy level. We use the normalised ergotropy
	\begin{equation}
		\varepsilon(t)=\frac{W(\hat\rho_B(t))}{N\omega_0}\in[0,1],
		\label{eq:eps}
	\end{equation}
	%
	and define the charging time to a target ergotropy $\varepsilon$ as
	\begin{equation}
		\tau^*(\varepsilon)=\inf\{t\ge0:\varepsilon(t)\ge\varepsilon\}.
		\label{eq:taustar}
	\end{equation}
	This is the first time at which $\varepsilon(t)$ attains the level
	$\varepsilon$, and $\tau^*$ is referred to as the charging time throughout.\\
	Since the dynamics generated by~\eqref{eq:dicke} from the initial
	state~\eqref{eq:init} is confined to the symmetric Dicke sector $J=N/2$, the
	ergotropy in Eq.~\eqref{eq:ergo} is the \emph{collective} ergotropy: the
	eigenvalues $\epsilon_k$ are those of $\hat H_B$ restricted to the
	$(N+1)$-dimensional symmetric subspace, and the extraction unitaries $\hat U$
	act within that sector. Throughout, $\varepsilon(t)$ therefore denotes the work
	extractable by collective unitary operations on the symmetric manifold. We
	evaluate the ergotropy with respect to the same control resource that defines
	the charging dynamics, rather than assuming arbitrary many-body unitaries that
	are absent from the model: a battery driven solely by $\hat J_x$ possesses no
	mechanism capable of implementing a generic unitary on the full
	$2^{N}$-dimensional space, so admitting such operations at the discharge stage would credit the device with work it cannot deliver.
	
	\emph{Resource counting.} The resources held fixed throughout are the
	coupling $\lambda$, the charger occupation $\bar n$ and the qubit frequency
	$\omega_0$, with $N$ varied; they enter the dynamics only through
	$\Gamma_N=2\lambda\sqrt{\bar n/N}$. Since the capacity $N\omega_0$ grows with
	$N$, we report the normalised ergotropy~\eqref{eq:eps}, and the implication
	for the collective advantage is discussed in Sec.~\ref{sec:conclusion}.

	\section{Main Result}
	\label{sec:main_result}
	
	The bound we prove has a geometric origin that is worth stating before the
	calculation. Under the charging dynamics the collective spin evolves as a unit
	Bloch vector; the ergotropy of the battery is a fixed function of the polar
	angle $\theta$ that this vector has turned through; and the drive limits the
	rate at which the angle can grow. The minimum charging time is then the
	required angle divided by the maximum angular velocity. The three steps below
	make each of these statements precise. Throughout,
	$|g\rangle\equiv|J,{-}J\rangle$ and $|e\rangle\equiv|J,{-}J{+}1\rangle$ denote
	the two lowest Dicke levels, separated by $\omega_0$, and
	$\GamN\equiv2\lambda\sqrt{\nbar/N}$; the supporting algebra is given in
	Appendix~\ref{app:A}.
	
	\subsection{Short-time behaviour}
	\label{sec:step1}
	
	\begin{lemma}[Short-time ergotropy]
		\label{lem:short_time}
		For the Dicke battery~\eqref{eq:dicke} with the initial
		state~\eqref{eq:init} at resonance $\omega_c=\omega_0$,
		\begin{equation}
			\eps(t)=A\,\lambda^2\nbar\,t^2+\mathcal{O}\!\left((\lambda t)^4\right),
			\qquad A=\frac{4}{N},
			\label{eq:eps_short}
		\end{equation}
		where the coefficient $A=4/N$ is independent of $\nbar$ and holds for every
		$\nbar\ge0$.
	\end{lemma}
	
	\begin{proof}
		The interaction of Eq.~\eqref{eq:dicke} first couples the initial state to
		$|e\rangle$; using $\Jx|g\rangle=(\sqrt N/2)|e\rangle$, the two factors of
		$\sqrt N$ cancel and
		\begin{equation}
			\Hint|g,\alpha\rangle=\lambda\,|e\rangle\otimes|\chi\rangle,
			\qquad |\chi\rangle\equiv(\hat a+\hat a^\dagger)|\alpha\rangle .
			\label{eq:Hintg}
		\end{equation}
		The coherent state is not an eigenstate of $\hat a+\hat a^\dagger$, so its
		first two moments enter the calculation independently,
		\begin{equation}
			\langle\alpha|\chi\rangle=2\sqrt{\nbar},
			\qquad
			\langle\chi|\chi\rangle=4\nbar+1 ,
			\label{eq:moments}
		\end{equation}
		and their difference is the vacuum contribution, which we retain exactly.
		To first order in $t$ the Dyson expansion gives
		$|\Psi(t)\rangle=|g,\alpha\rangle-i\lambda t|e,\chi\rangle
		+\mathcal{O}(t^2)$, the second-order term supplying the normalisation
		deficit and a population of $|J,{-}J{+}2\rangle$ of order $(\lambda t)^4$
		[Appendix~\ref{app:A}, Eqs.~\eqref{eq:A2} and~\eqref{eq:Ad}]. Tracing out
		the charger, the battery is described to this order by the $2\times2$ block
		\begin{equation}
			\rhoB=\begin{pmatrix}1-p & \kappa^{*}\\[3pt]\kappa & p\end{pmatrix},
			\quad
			\begin{aligned}
				p&=\lambda^2t^2\langle\chi|\chi\rangle,\\
				\kappa&=-i\lambda t\langle\alpha|\chi\rangle,
			\end{aligned}
			\label{eq:block}
		\end{equation}
		so that the excited population carries the second moment of the quadrature
		and the coherence the first. The mean energy measured from the ground level
		is $E=p\,\omega_0$. The ergotropy is $E$ minus the energy of the passive
		state associated with $\rhoB$; from $\mathrm{Tr}\,\rhoB=1$ and
		$\det\rhoB=p(1-p)-|\kappa|^2$, the smaller eigenvalue is
		\begin{equation}
			r_-=p-|\kappa|^2
			=\lambda^2t^2\!\left[\langle(\hat a+\hat a^\dagger)^2\rangle
			-\langle\hat a+\hat a^\dagger\rangle^2\right]=\lambda^2t^2 ,
			\label{eq:rminus}
		\end{equation}
		the bracket being the variance of a coherent-state quadrature, equal to
		unity. The passive state places $r_-$ on the upper level, so its energy
		above the ground level is $r_-\omega_0$, and
		\begin{equation}
			\W=\left(p-r_-\right)\omega_0=|\kappa|^2\omega_0
			=4\lambda^2\nbar\,t^2\,\omega_0 .
			\label{eq:Wshort}
		\end{equation}
		Dividing by $N\omega_0$ gives Eq.~\eqref{eq:eps_short}. The odd orders
		vanish by the time-reversal symmetry of the real Hamiltonian
		[Appendix~\ref{app:A}, Eq.~\eqref{eq:Ag}], so the first correction is
		$\mathcal{O}((\lambda t)^4)$.
	\end{proof}
	
	The structure of Eqs.~\eqref{eq:rminus} and~\eqref{eq:Wshort} distinguishes
	three quantities that are otherwise easily conflated. The mean energy stored is
	governed by the second moment $\langle(\hat a+\hat a^\dagger)^2\rangle=4\nbar+1$;
	the passive-state energy by the quadrature variance, which is unity; and the
	extractable work by the squared first moment
	$\langle\hat a+\hat a^\dagger\rangle^2=4\nbar$. The vacuum contribution enters
	the first two identically and therefore cancels from the difference: at leading
	order the ergotropy is carried by the coherent component of the drive alone,
	and the incoherent part deposits energy in the battery without depositing work.
	Were this not so the coefficient would read
	$A=4(\nbar+\tfrac14)/(N\nbar)$, which the numerics of
	Sec.~\ref{sec:numerics} exclude. The factor $1/N$ follows from the cancellation
	in Eq.~\eqref{eq:Hintg}, the excitation amplitude being $N$-independent while
	the normalisation of the ergotropy scales with $N$.
	
	\subsection{Bound on the rotation rate}
	\label{sec:step2}
	
	We now pass to the classical-field limit, in which the charger is undepleted
	and its quadrature is replaced by the c-number
	$\hat a+\hat a^\dagger\to2\sqrt{\nbar}\cos\omega_ct$. The scope of this
	replacement is stated precisely: it holds when the change in the charger over a
	charging time is negligible, i.e.\ for $\nbar$ large at fixed $\GamN$, and the
	finite-$\nbar$ corrections are quantified numerically in
	Sec.~\ref{sec:numerics}.
	
	\begin{lemma}[Envelope for the ergotropy]
		\label{prop:global}
		In the classical-field limit the polar angle $\theta$ of the collective
		Bloch vector satisfies $|\dot\theta(t)|\le2\GamN$ for any detuning
		$\Delta=\omega_c-\omega_0$, and consequently
		\begin{equation}
			\eps(t)\le\sin^2\!\left(\GamN t\right),
			\qquad 0\le t\le\frac{\pi}{2\GamN}.
			\label{eq:global_bound}
		\end{equation}
	\end{lemma}
	
	\begin{proof}
		With the replacement above and
		$(2\lambda/\sqrt N)\cdot2\sqrt{\nbar}=2\GamN$, the Hamiltonian is
		\begin{equation}
			\hat H_{\rm cl}(t)=\omega_0\Jz+2\GamN\cos(\omega_ct)\,\Jx
			\equiv\bm\omega(t)\cdot\bm{\hat J},
			\label{eq:Hcl}
		\end{equation}
		with $\bm\omega(t)=(2\GamN\cos\omega_ct,0,\omega_0)$. This is linear in the collective spin, so the Heisenberg equations
		close and $\hat H_{\mathrm{cl}}$ generates a rigid $SU(2)$ rotation:
		\begin{equation}
			\frac{d}{dt}\langle\hat{\mathbf J}\rangle
			= \boldsymbol\omega(t)\times\langle\hat{\mathbf J}\rangle
			\quad\Longrightarrow\quad
			\frac{d}{dt}\bigl|\langle\hat{\mathbf J}\rangle\bigr|^{2}
			= 2\,\langle\hat{\mathbf J}\rangle\cdot
			\bigl(\boldsymbol\omega\times\langle\hat{\mathbf J}\rangle\bigr)=0,
			\label{eq:norm-conserved}
		\end{equation}
		so $|\langle\hat{\mathbf J}\rangle|=J$ is conserved and the initial spin
		coherent state remains spin coherent. The unit vector
		$\mathbf r=\langle\hat{\mathbf J}\rangle/J$ obeys
		$\dot{\mathbf r}=\boldsymbol\omega\times\mathbf r$, whose $z$-component,
		\begin{equation}
			\dot r_z=\omega_xr_y-\omega_yr_x=\omega_xr_y ,
			\label{eq:rzdot}
		\end{equation}
		contains no $\omega_z$ because $\omega_y=0$ and the $z$-component of a cross
		product is independent of the $z$-components of its factors. As the
		ergotropy of a pure state depends on $\bm r$ only through $r_z$, the free
		frequency $\omega_0$ does not enter it, and no rotating frame or
		rotating-wave approximation is invoked. Writing $r_z=-\cos\theta$ gives
		$\dot r_z=\sin\theta\,\dot\theta$, and $|\bm r|=1$ gives
		$|r_y|\le\sin\theta$; with $|\cos\omega_ct|\le1$,
		\begin{equation}
			\sin\theta\,|\dot\theta|\le2\GamN\sin\theta
			\;\Longrightarrow\;
			\theta(t)\le2\GamN t .
			\label{eq:sinbnd}
		\end{equation}
		The state is pure throughout, so its passive state is the ground state and
		\begin{equation}
			\eps=\tfrac12\left(1-\cos\theta\right)=\sin^2(\theta/2),
			\label{eq:eps_theta}
		\end{equation}
		which increases monotonically on $\theta\in[0,\pi]$. Inserting
		$\theta(t)\le2\GamN t$ gives Eq.~\eqref{eq:global_bound}; the restriction
		$2\GamN t\le\pi$ is the range over which this monotonicity is used. The
		relation $\eps=\sin^2(\theta/2)$ is the pure-state, classical-field
		identity; at finite $\nbar$ the reduced state $\rhoB$ is mixed, its passive
		state is no longer the ground state, and $\eps$ falls below
		$\sin^2(\theta/2)$, the shortfall being quantified numerically in
		Sec.~\ref{sec:numerics}.
	\end{proof}
	
	Two properties of this result should be noted. The relation~\eqref{eq:eps_theta}
	between ergotropy and angle is kinematic, containing no reference to the
	Hamiltonian, while the bound~\eqref{eq:sinbnd} on $\dot\theta$ is dynamical and
	contains no reference to the ergotropy; the two combine to give the envelope.
	Because only the amplitude of the drive enters~\eqref{eq:sinbnd}, the bound is
	insensitive to detuning, and any control field $f(t)\Jx$ with
	$|f(t)|\le2\GamN$ obeys the same envelope. Expanding
	Eq.~\eqref{eq:global_bound} at short times returns
	$(4/N)\lambda^2\nbar t^2$, in agreement with Lemma~\ref{lem:short_time}; the
	envelope is thus tangent to the exact trajectory at the origin, a property that
	underlies the near-saturation observed numerically. We stress that
	Eq.~\eqref{eq:global_bound} is derived in the classical-field limit; at finite
	$\nbar$ the charger quadrature has unbounded norm and the argument does not
	close, and the numerical results of Sec.~\ref{sec:numerics} show that in that
	case the ergotropy remains below the envelope.
	
	\subsection{The speed limit}
	\label{sec:step3}
	
	\begin{theorem}[Ergotropy speed limit]
		\label{thm:QSL}
		For any $\lambda>0$, $N\ge2$ and $\eps\in(0,1]$, in the classical-field
		limit,
		\begin{equation}
			\taus(\eps)\;\ge\;\tauQSL
			=\frac{\sqrt N\,\arcsin\sqrt\eps}{2\lambda\sqrt{\nbar}} .
			\label{eq:QSL_main}
		\end{equation}
	\end{theorem}
	
	\begin{proof}
		Suppose first $\taus\le\pi/(2\GamN)$. Then Eq.~\eqref{eq:global_bound}
		applies at $t=\taus$, where $\eps(\taus)=\eps$, giving
		$\sqrt\eps\le\sin(\GamN\taus)$ with $\GamN\taus\in[0,\pi/2]$; since
		$\arcsin$ is increasing on $[0,1]$, this rearranges to
		$\GamN\taus\ge\arcsin\sqrt\eps$. If instead $\taus>\pi/(2\GamN)$, then
		$\GamN\taus>\pi/2\ge\arcsin\sqrt\eps$ for every $\eps\in(0,1]$, and the
		inequality holds \emph{a fortiori}. The two cases exhaust the possibilities.
	\end{proof}
	
	\begin{corollary}[Collapse variable]
		\label{cor:universal}
		With $X\equiv\GamN\taus(\eps)$, Theorem~\ref{thm:QSL} reads
		\begin{equation}
			X\;\ge\;\arcsin\sqrt\eps\;\ge\;\sqrt\eps
			\label{eq:universal}
		\end{equation}
		for all $(\lambda,\nbar,N,\eps)$, the second inequality following from
		$\arcsin x\ge x$ on $[0,1]$.
	\end{corollary}
	
	The bound admits a direct reading on the Bloch sphere. The chord joining the
	initial and final states has length $2\sin(\theta/2)=2\sqrt\eps$, and the arc
	joining them has length $\theta=2\arcsin\sqrt\eps$; the system evolves along the
	arc, so it is the arc that sets the minimum time. The chord provides the weaker
	estimate $\sqrt\eps$, which coincides with the arc as $\eps\to0$ and lies below
	it by a growing margin as $\eps$ increases. In this reading $\GamN$ is half the
	maximum angular velocity of the collective vector, and Eq.~\eqref{eq:QSL_main}
	states that the charging time is at least the arc length divided by that
	velocity: this is the structure of the Mandelstam--Tamm relation, with the
	geodesic angle replaced by the ergotropy angle and the energy uncertainty by
	$\GamN$.
	\providecommand{\eps}{\varepsilon}
	\providecommand{\nbar}{\bar{n}}
	\providecommand{\taus}{\tau^{*}}
	\providecommand{\GamN}{\Gamma_N}

	\section{Numerical Verification}
	\label{sec:numerics}
	
	The analytic bound of Sec.~\ref{sec:main_result} is derived in the
	classical-field limit. To assess its status at finite charger occupation we
	integrate the full Dicke Hamiltonian~\eqref{eq:dicke}. We
	diagonalise $\HD$ exactly in the symmetric sector, truncating the cavity at
	$N_{\rm Fock}$ Fock states, and evaluate the ergotropy~\eqref{eq:ergo} at each
	time from the passive state of $\hat\rho_B(t)=\mathrm{Tr}_C\,\hat\rho_{\rm
		tot}(t)$; no approximation to the reduced state is made. The scan covers
	$N=2,3,4,5$, the coupling $\lambda/\omega_0\in[0.1,2.0]$, the occupation
	$\nbar\in[1,20]$, and $50$ target ergotropies $\eps\in[0.02,0.96]$.
	Throughout the scan the charger remains a quantised mode, so depletion,
	charger--battery correlations and the mixedness of $\rhoB$ are all retained;
	the data therefore test the bound at finite $\nbar$, beyond the validity of
	the classical-field replacement used in Lemma~\ref{prop:global}.
	Figure~\ref{fig:collapse} is organised accordingly: $\eps(t)$ for
	representative $(\lambda/\omega_0,\nbar)$~(a); the admitted charging times of
	the $N=2$ scan in the collapse plane $(\eps,X=\GamN\taus)$~(b); and those
	same data at fixed $\eps=0.5$, as $\taus$ versus $1/\GamN$~(c).
		\begin{figure*}[t]
		\centering
		\includegraphics[width=1\textwidth]{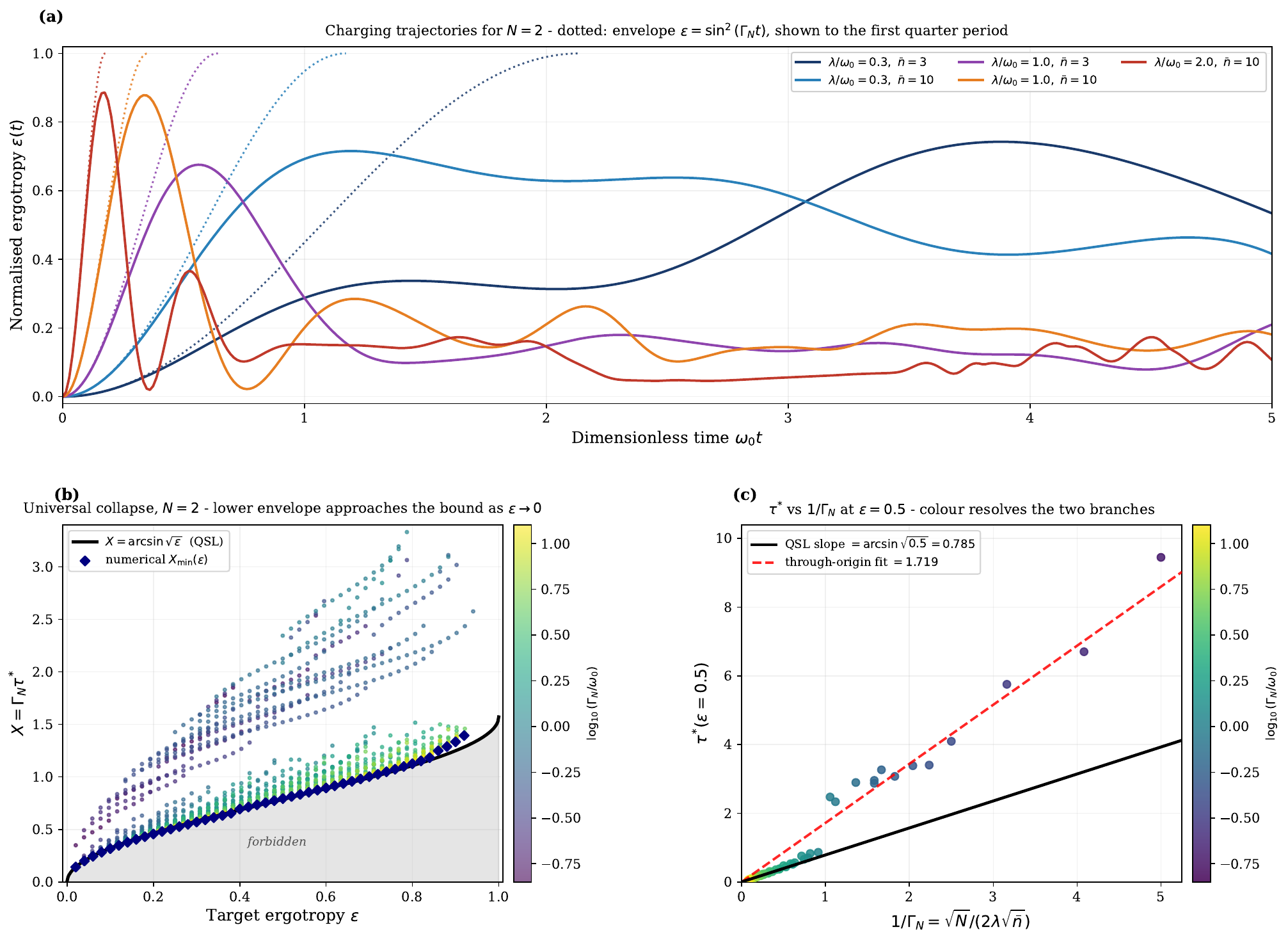}
		\caption{%
			(Color online) Ergotropy charging of the $N=2$ Dicke battery.
			\textbf{(a)}~Charging dynamics: the normalised ergotropy $\eps(t)$ for
			several couplings and charger occupations $(\lambda/\omega_0,\nbar)$.
			Solid curves: exact numerics; dotted: the classical-field envelope
			$\sin^2(\GamN t)$, shown up to the first quarter period where it applies.
			Each trajectory rises below the envelope, reaches a first maximum, and
			then oscillates as energy is exchanged with the charger.
			\textbf{(b)}~Collapse of the charging times: every point is plotted
			as $X=\GamN\taus$ against the target ergotropy $\eps$, coloured by
			$\log_{10}(\GamN/\omega_0)$. Solid curve: the bound $X=\arcsin\sqrt\eps$
			(Corollary~\ref{cor:universal}); filled diamonds: the numerical lower
			envelope $X_{\min}(\eps)$; the shaded region below the curve is
			dynamically inaccessible.
			\textbf{(c)}~The same data at fixed $\eps=0.5$, plotted as $\taus$
			against $1/\GamN$. The speed limit predicts a line through the origin of
			slope $\arcsin\sqrt{0.5}$ (solid); the fast-branch points
			($\GamN\gtrsim\omega_0$) follow it, while the slow-branch points
			($\GamN\ll\omega_0$) lie on a second ray of about twice the slope, the
			factor of two being the halved rotation rate discussed in the text.
			Colour encodes $\log_{10}(\GamN/\omega_0)$.
			Parameters throughout: $\omega_c=\omega_0$, $\hbar=1$,
			$N_{\rm Fock}=60$ and $\GamN=2\lambda\sqrt{\nbar/N}$; $\taus$ obtained
			by root-finding.}
		\label{fig:collapse}
	\end{figure*}
	Two points of numerical procedure are worth stating, as both affect the margins
	reported below. First, the charging time~\eqref{eq:taustar} is obtained by
	locating the crossing $\eps(t)=\eps$ through bracketing and bisection, rather
	than by reading the nearest point of a fixed time grid; the latter returns a
	value no smaller than the true $\taus$ and can therefore convert a saturated
	bound into an apparent margin. Grid reading inflates the smallest observed
	margin by an order of magnitude, an effect purely of discretisation; the
	bracket is refined by bisection to machine precision in $t$. Second, a
	target is admitted only when it is reached within the integration window,
	$\eps\le\max_t\eps(t)$; targets exceeding the attainable maximum, which occurs
	for small $\nbar$, are left undefined rather than assigned the window boundary.
	A violation is registered whenever
	$\GamN\taus<(1-10^{-6})\arcsin\sqrt\eps$. At every system size the minimum
	ratio is attained at the strongest drive and the smallest target of the scan,
	consistent with the envelope being tightest as $\eps\to0$.
	The reported margins are robust against Fock truncation.
	Table~\ref{tab:convergence} lists $\min X/\!\arcsin\!\sqrt\eps$ at the most
	demanding occupation $\nbar=20$ ($N=2$) for $N_{\rm Fock}=40,60,80$; the ratio
	is unchanged between $60$ and $80$, so the zero-violation result reflects the
	physics and not the cutoff.
		\begin{figure*}[t]
		\centering
		\includegraphics[width=0.9\textwidth]{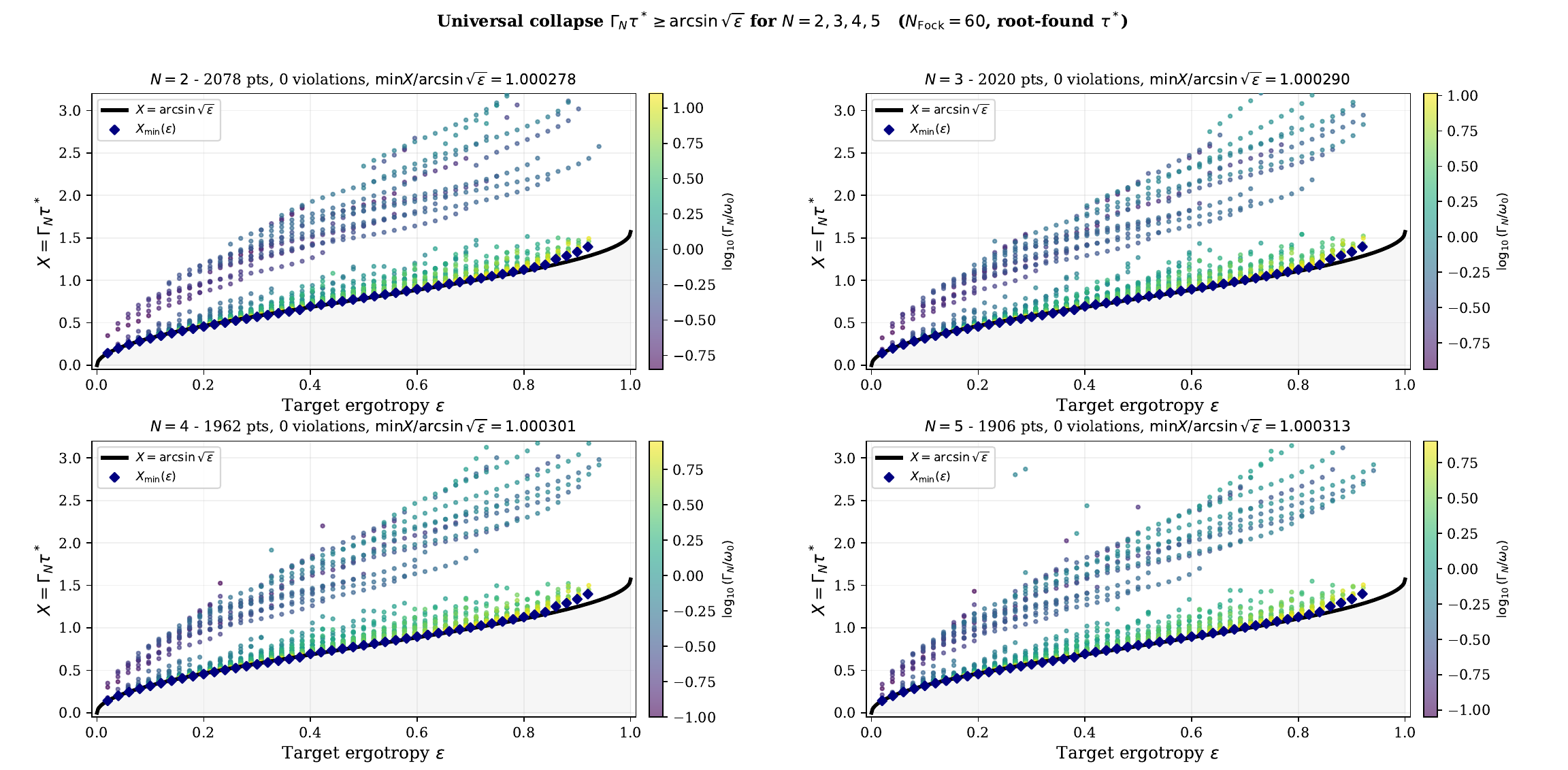}
		\caption{%
			(Color online) Collapse $X=\GamN\taus\ge\arcsin\sqrt\eps$ for
			$N=2,3,4,5$. Each panel shows all admitted charging times in the
			$(\eps,X)$ plane, coloured by $\log_{10}(\GamN/\omega_0)$; solid curve:
			the bound; diamonds: numerical lower envelope; shaded: inaccessible
			region. No violation occurs at any system size.
			Parameters: $\omega_c=\omega_0$, $\hbar=1$, $N_{\rm Fock}=60$,
			$\lambda/\omega_0\in[0.1,2.0]$, $\nbar\in[1,20]$ and
			$\GamN=2\lambda\sqrt{\nbar/N}$; $\taus$ obtained by root-finding.}
		\label{fig:collapse_all}
	\end{figure*}
	
	\begin{table}[t]
		\caption{Convergence of the QSL margin under Fock truncation, at the most
			demanding occupation $\nbar=20$ ($N=2$). The minimum ratio
			$\min X/\!\arcsin\!\sqrt\eps$ is unchanged between $N_{\rm Fock}=60$ and
			$80$; the production value is therefore converged.}
		\label{tab:convergence}
		\begin{ruledtabular}
			\begin{tabular}{ccc}
				$N_{\rm Fock}$ & Admitted points & $\min X/\!\arcsin\!\sqrt\eps$ \\
				\hline
				40 & 373 & 1.000333 \\
				60 & 372 & 1.000278 \\
				80 & 372 & 1.000278 \\
			\end{tabular}
		\end{ruledtabular}
	\end{table}
	We first test Lemma~\ref{lem:short_time}. Table~\ref{tab:Anum} lists the
	coefficient $A_{\rm num}=\eps(t)/(\lambda^2\nbar t^2)$ measured at a fixed early
	time for $N=2$--$5$. It reproduces $A=4/N$ across the full range of $\nbar$,
	with a residual that grows with $\lambda$ and is consistent in magnitude with
	the $\mathcal{O}((\lambda t)^4)$ correction of the Lemma. Extrapolating
	$A_{\rm num}$ to $t\to0$ removes this residual and returns $4/N$ to the
	numerical precision of the extrapolation, independently of $\nbar$; in
	particular the value $A=4(\nbar+\tfrac14)/(N\nbar)$ that would obtain were the
	charger vacuum to contribute extractable work is excluded already at $\nbar=1$,
	where it differs from $4/N$ by $25\%$. This confirms the cancellation
	identified in Eq.~\eqref{eq:rminus}: at leading order the vacuum quadrature
	deposits energy in the battery but no ergotropy.
	
	\begin{table}[t]
		\caption{Short-time coefficient
			$A_{\rm num}=\eps(t)/(\lambda^2\nbar t^2)$ evaluated at
			$t=0.1\,\omega_0^{-1}$, compared with the analytic value $A=4/N$ of
			Lemma~\ref{lem:short_time} ($A_{\rm th}=2,\tfrac43,1,\tfrac45$ for
			$N=2,3,4,5$). The residual increases with $\lambda$, consistent with the
			$\mathcal{O}((\lambda t)^4)$ correction.}
		\label{tab:Anum}
		\begin{ruledtabular}
			\begin{tabular}{cccccc}
				$\lambda/\omega_0$ & $\nbar$ & $N{=}2$ & $N{=}3$ & $N{=}4$ & $N{=}5$\\
				\hline
				0.1 &  1 & 1.9914 & 1.3277 & 0.9958 & 0.7966 \\
				0.1 & 10 & 1.9902 & 1.3271 & 0.9955 & 0.7964 \\
				0.3 & 10 & 1.9780 & 1.3217 & 0.9924 & 0.7945 \\
				0.5 & 10 & 1.9540 & 1.3110 & 0.9864 & 0.7906 \\
				1.0 & 10 & 1.8445 & 1.2618 & 0.9586 & 0.7728 \\
			\end{tabular}
		\end{ruledtabular}
	\end{table}

	Panel~(a) of Fig.~\ref{fig:collapse} shows the charging dynamics directly, as
	the ergotropy $\eps(t)$ for several representative couplings and charger
	occupations. Each curve rises from zero, reaches a first maximum, and then
	oscillates: the battery does not charge monotonically but exchanges energy back
	and forth with the charger, so that the useful quantity is the charging time
	to a target and not the long-time behaviour. Over the initial rise every
	curve stays below the envelope $\sin^2(\GamN t)$, drawn dotted for each case up
	to the first quarter period where the envelope applies; the more strongly
	driven cases rise faster and hug the envelope closely, while the weakly driven
	cases lag behind it. This is the same two-regime structure that reappears
	quantitatively below. The panel also shows why the classical field is only a
	leading description: at finite $\nbar$ the curves fall increasingly short of the
	envelope as time advances, the charger being depleted and correlated with the
	battery, which is precisely the regime the analytic bound does not resolve and
	the numerics do.
	
	We next test the bound itself. Writing $X=\GamN\taus$, all charging times
	fall on a common locus when plotted against $\eps$, independent of the
	individual values of $\lambda$ and $\nbar$; the collapse is shown for $N=2$ in
	Fig.~\ref{fig:collapse}(b) and for $N=2$--$5$ in Fig.~\ref{fig:collapse_all}. In
	every case the data lie on or above the curve $X=\arcsin\sqrt\eps$. The
	statistics are summarised in Table~\ref{tab:min_ratio}: across the four system
	sizes the bound is not violated at any of the $7966$ admitted points, and the
	smallest ratio $X/\arcsin\sqrt\eps$ approaches unity as $\eps\to0$, the regime
	in which the classical-field envelope is tightest. The numerical lower envelope
	$X_{\min}(\eps)$ follows the analytic curve across the accessible range of
	$\eps$ rather than only at small $\eps$; we emphasise that this near-saturation
	is an empirical observation over the scanned parameters, the analytic result
	being an inequality rather than an equality.
	
	That the lower envelope is attained, to within the resolution of the scan, at
	every system size is consistent with the collective rotation of Sec.~\ref{sec:main_result}
	being the fastest charging path available at fixed $(\lambda,\nbar,\omega_0)$:
	no parameter set in our scan reaches a given ergotropy in less than
	$\arcsin\sqrt\eps/\GamN$.
	
	The bound also survives the alternative convention discussed in
	Sec.~\ref{sec:model}, in which the passive state is built from the full $2^N$
	spectrum with its level degeneracies. That convention lowers the passive
	energy and shortens $\taus$, so it is the stronger test; repeating the scan for
	$N=2$ and $N=3$ leaves the bound unviolated, with minimum ratios differing from
	the collective ones in the sixth decimal (Table~\ref{tab:min_ratio}). The
	agreement is expected from Lemma~\ref{lem:short_time}: where the bound is
	closest to saturation the battery state is supported on $\{|g\rangle,
	|e\rangle\}$, where no degeneracy is available and the two passive-state
	constructions coincide.

	\begin{table}[t]
		\caption{Collapse statistics for $N=2$--$5$, with $\taus$ obtained by
			root-finding over $\lambda/\omega_0\in[0.1,2.0]$ and $\nbar\in[1,20]$.
			The bound $X\ge\arcsin\sqrt\eps$ is not violated at any admitted point;
			the minimum ratio is reported and approaches unity as $\eps\to0$. The
			last column repeats the test with the passive state built from the full
			$2^N$ spectrum, which is the stronger test.}
		\label{tab:min_ratio}
		\begin{ruledtabular}
			\begin{tabular}{ccccc}
				$N$ & Admitted points & Violations & $\min X/\!\arcsin\!\sqrt\eps$
				& full $2^N$ \\
				\hline
				2 & 2078 & 0 & 1.000278 & 1.000277 \\
				3 & 2020 & 0 & 1.000290 & 1.000287 \\
				4 & 1962 & 0 & 1.000301 & \\
				5 & 1906 & 0 & 1.000313 & \\
				\hline
				Total & 7966 & 0 & & \\
			\end{tabular}
		\end{ruledtabular}
	\end{table}
	
	Finally we record the behaviour away from saturation, which fixes the
	interpretation of the scatter above the envelope in
	Figs.~\ref{fig:collapse} and~\ref{fig:collapse_all} and is made explicit in
	panel~(c) of Fig.~\ref{fig:collapse}. Two regimes are present, set by whether
	the drive rate $\GamN$ is large or small compared with the qubit frequency
	$\omega_0$. When $\GamN\gtrsim\omega_0$ the counter-rotating terms
	of~\eqref{eq:dicke} act on the timescale of the rotation itself, the collective
	vector turns at the full rate $2\GamN$, and the ergotropy follows the envelope
	$\sin^2(\GamN t)$; the charging time then approaches the bound. When
	$\GamN\ll\omega_0$ the counter-rotating terms average out and the effective
	rotation rate is halved, so the same ergotropy takes about twice as long to
	reach.
	
	Panel~(c) isolates this at fixed $\eps=0.5$ by plotting $\taus$ against
	$1/\GamN$. Because $\GamN$ is the only rate in the problem, the speed limit
	predicts a straight line through the origin whose slope is
	$\arcsin\sqrt{0.5}$, and the fast-branch data indeed follow it: these are the
	points that saturate the bound. The slow-branch data fall on a second line
	through the origin of roughly twice the slope, the factor of two being the
	halved rotation rate rather than any departure from the limit. The offset
	between the two families in panel~(c) is therefore physical, and both are
	governed by the single rate $\GamN$.
	
	\section{Discussion and Conclusion}
	\label{sec:conclusion}

	We have derived and proved a quantum speed limit for ergotropy charging
	in the $N$-qubit Dicke battery,
	$\Gamma_N\tau^{*}(\varepsilon)\ge\arcsin\sqrt{\varepsilon}$, proved in the
	classical-field limit for all $\lambda$, $N$ and $\varepsilon\in(0,1]$ and
	verified numerically at finite charger occupation with no violations, the
	bound being approached most closely in the limit of small target work. The composite rate $\Gamma_N=2\lambda\sqrt{\bar n/N}$ appears as
	the amplitude of the classical drive, which limits the angular velocity of the
	collective Bloch vector; together with the exact relation
	$\varepsilon=\sin^{2}(\theta/2)$ this yields the envelope
	$\varepsilon(t)\le\sin^{2}(\Gamma_N t)$ from which the bound follows. The
	$\sqrt{\bar n}$ dependence reflects coherent enhancement: for a thermal charger
	of the same mean energy the first moment
	$\langle\hat a+\hat a^{\dagger}\rangle$ vanishes, so by
	Lemma~\ref{lem:short_time} no ergotropy is generated at leading order, and the
	charging is coherent in origin rather than merely faster. The bound is in
	direct analogy with the Mandelstam--Tamm relation~\cite{Mandelstam1945},
	applied here to a thermodynamic observable rather than to state
	distinguishability.
	
	It is worth being explicit about what is, and is not, elementary in this
	result. The underlying dynamics is semiclassical spin rotation, and in the
	strong-coupling regime the trajectory is available in closed form. Our claim is
	narrower and concerns work extraction rather than the rotation itself. The
	ergotropy is defined by an optimisation over unitaries, yet it collapses onto
	the single angle $\varepsilon=\sin^{2}(\theta/2)$, with the charger vacuum
	depositing energy but no extractable work; this reduction is not part of the
	rotation problem and is what allows a kinematic bound on the angle to become a
	thermodynamic bound on the charging time. The bound is derived in the
	classical-field limit; at finite $\bar n$ the charger is a dynamical mode and
	the reduced state is mixed, and there the inequality is not proved but tested,
	the numerics of Sec.~\ref{sec:numerics} finding it respected at every admitted
	point. That the numerical lower envelope approaches the bound at every system
	size, rather than merely respecting it, suggests that collective rotation is
	the fastest charging path among the parameter sets spanned by our scan: none
	of those examined reaches a given ergotropy in a shorter time.
	
	The rate $\Gamma_N=2\lambda\sqrt{\bar n/N}$ decreases as $1/\sqrt N$ at fixed
	$\lambda$ and $\bar n$, so the time to a fixed \emph{normalised} ergotropy
	grows as $\sqrt N$. This is consistent with the superextensive $\sqrt N$
	charging power of the Dicke battery: the capacity grows as $N$ while the
	normalised rate falls as $1/\sqrt N$, so the total power
	$\sim N\omega_0\Gamma_N\propto\sqrt N$; the collective advantage resides in the
	power, not in the normalised charging speed.
	
	The bound also carries a practical consequence for circuit-QED
	implementations, where it fixes the \emph{necessary} coupling to charge within
	a target time, $\lambda\gtrsim\sqrt N\,\arcsin\sqrt{\varepsilon}\,/(2\tau_{\rm
		target}\sqrt{\bar n})$. For representative parameters, namely
	$\varepsilon=0.8$, $\tau_{\rm target}=100\,$ns, $\bar n=10$, $N=2$ and
	$\omega_0/2\pi=5\,$GHz, this places the required coupling well within the range
	of present architectures~\cite{Krantz2019,Blais2020}, so that the speed limit
	is not a practical obstacle but a design target.
	
	Several features distinguish the bound from earlier results. It is
	model-specific and analytically explicit in the experimentally accessible
	parameters $(\lambda,\bar n,N)$, whereas general observable speed
	limits~\cite{Shrimali2024} and charging-distance
	bounds~\cite{Gyhm2024} are model-independent and are not saturated for a given
	system. Instantaneous power bounds~\cite{Gyhm2022,GarciaPintos2020,Shrimali2024}
	constrain the rate of energy storage but do not by themselves yield a minimum
	time to a prescribed ergotropy, which requires integrating a
	trajectory-dependent inequality; the present bound supplies that time in closed
	form; the comparison is carried out explicitly in
	Appendix~\ref{app:powerbound} and Fig.~\ref{fig:comparison}, where the ordering
	$\arcsin\sqrt{\varepsilon}\ge\sqrt{\varepsilon}\ge\varepsilon$ is established.
	The commutation-matrix analysis of Ref.~\cite{GyhmaFischer2024} isolates
	the entanglement contribution to charging power in a model-independent way, and
	our result is complementary, giving a saturated bound for the Dicke model in
	terms of measurable quantities.
	
	The single substantive assumption is the classical-field limit. At finite
	$\bar n$ the charger quadrature is an operator of unbounded norm and the
	angular-velocity argument does not close, though the numerics indicate that its
	fluctuations only ever slow the charging. A rigorous finite-$\bar n$ statement,
	presumably with $\Gamma_N$ replaced by a fluctuation-corrected rate, is the
	natural next step. Because the derivation uses only that the drive amplitude is
	bounded, the same limit applies to arbitrary pulse-shaping protocols at fixed
	amplitude, and it extends naturally to open-system
	dynamics~\cite{Funo2019,Barra2022} and to the Tavis--Cummings and
	Jaynes--Cummings models. An experimental test in a superconducting circuit,
	where $\Gamma_N$ is set directly by the coupling and the drive, appears within
	current reach.
	
	\section*{Data availability statement}
	The data that support the findings of this study are available upon reasonable
	request from the authors.
	
	\section*{Acknowledgments}
	A.J. gratefully acknowledges the support of the Department of Optics,
	Palack\'{y} University Olomouc, Czech Republic, where part of this
	work was carried out.
	A.E.A. completed part of this work during a research visit to the
	LPMMC, CNRS, in Grenoble, France. He extends his sincere gratitude
	to the CNRS -- F\'{e}d\'{e}ration de Recherche QuantAlps, Comit\'{e}
	de Direction QuantAlps, for their financial support and for fostering
	a stimulating and friendly research environment.
		\appendix
		
		\section{Supporting calculations}
		\label{app:A}
		
		This appendix supplies the algebra used in Sec.~\ref{sec:main_result}. The
		material is elementary but is included in full so that every step of the two
		proofs can be checked independently. Parts (a)--(g) support
		Lemma~\ref{lem:short_time}; parts (h)--(j) support
		Lemma~\ref{prop:global} and Theorem~\ref{thm:QSL}. Throughout,
		$|g\rangle=|J,{-}J\rangle$, $|e\rangle=|J,{-}J{+}1\rangle$, and
		$H_0=\omega_0\Jz+\omega_c\hat a^\dagger\hat a$.
		
		\emph{a. Collective matrix elements.}
		With $J=N/2$ and
		$\hat J_\pm|J,m\rangle=\sqrt{J(J+1)-m(m\pm1)}\,|J,m\pm1\rangle$, the two lowest
		levels satisfy
		\begin{equation}
			\hat J_+|g\rangle=\sqrt{2J}\,|e\rangle=\sqrt N\,|e\rangle,
			\qquad
			\hat J_-|g\rangle=0 ,
			\label{eq:A1}
		\end{equation}
		so that $\Jx|g\rangle=\tfrac12(\hat J_++\hat J_-)|g\rangle=(\sqrt N/2)|e\rangle$.
		The interaction of Eq.~\eqref{eq:dicke} therefore produces
		Eq.~\eqref{eq:Hintg}, the factor $2/\sqrt N$ of the Hamiltonian cancelling the
		$\sqrt N/2$ of the matrix element; the excitation amplitude is thus independent
		of $N$.
		
		\emph{b. Interaction picture and Dyson series.}
		Writing $\HD=H_0+\lambda\hat V$ with
		$\hat V=(2/\sqrt N)(\hat a+\hat a^\dagger)\Jx$, the interaction-picture operator
		$\hat V_I(t)=e^{iH_0t}\hat Ve^{-iH_0t}$ is
		\begin{equation}
			\hat V_I(t)=\frac{2}{\sqrt N}
			\left(\hat a\,e^{-i\omega_ct}+\hat a^\dagger e^{i\omega_ct}\right)
			\left(\Jx\cos\omega_0t-\Jy\sin\omega_0t\right),
			\label{eq:A2}
		\end{equation}
		and the evolution to second order is
		\begin{multline}
			|\Psi(t)\rangle = e^{-iH_0t}\!\left[
			\mathbb{1}-i\lambda\!\int_0^t\!\hat V_I(s)\,ds\right.\\
			\left.-\lambda^2\!\!\int_0^t\!\!\!\int_0^s\!\hat V_I(s)\hat V_I(u)\,du\,ds
			\right]\!|\Psi(0)\rangle+\mathcal{O}(\lambda^3).
			\label{eq:A3}
		\end{multline}
		Over times $t\ll\omega_0^{-1}$ the oscillating factors in
		Eq.~\eqref{eq:A2} may be set to unity, since their expansion contributes to
		$\eps$ only at $\mathcal{O}((\lambda t)^4)$; this yields the first-order state
		of Lemma~\ref{lem:short_time}.
		
		\emph{c. Vanishing of the linear term.}
		Using $[\Jz,\Jx]=i\Jy$ and $\langle g|\Jy|g\rangle=0$, one has
		\begin{equation}
			\left.\frac{d\langle\Jz\rangle}{dt}\right|_{t=0}
			=i\lambda\langle\Psi(0)|[\hat V,\Jz]|\Psi(0)\rangle=0 ,
			\label{eq:A4}
		\end{equation}
		so the ergotropy has no term linear in $t$ and grows as $t^2$.
		
		\emph{d. Truncation to the two-level block.}
		From
		$\Jx|e\rangle=\tfrac12\sqrt N\,|g\rangle
		+\tfrac12\sqrt{2N-2}\,|J,{-}J{+}2\rangle$,
		the second-order term of Eq.~\eqref{eq:A3} contains a component in
		$|J,{-}J{+}2\rangle$ with amplitude $\mathcal{O}(\lambda^2t^2)$, hence
		population $\mathcal{O}((\lambda t)^4)$. To the order retained the state is
		confined to $\{|g\rangle,|e\rangle\}$, and its $|g\rangle$ component acquires
		the second-order correction
		\begin{equation}
			-\tfrac12\lambda^2t^2\langle(\hat a+\hat a^\dagger)^2\rangle ,
			\label{eq:Ad}
		\end{equation}
		which is the depletion $1-p$ of Eq.~\eqref{eq:block}.
		
		\emph{e. Quadrature moments.}
		For $\alpha=\sqrt{\nbar}$ real,
		\begin{equation}
			\langle\hat a+\hat a^\dagger\rangle=2\sqrt{\nbar},
			\qquad
			\langle(\hat a+\hat a^\dagger)^2\rangle=4\nbar+1 ,
			\label{eq:A5}
		\end{equation}
		the second following from $\hat a|\alpha\rangle=\alpha|\alpha\rangle$ and
		$[\hat a,\hat a^\dagger]=1$. The variance
		$\langle(\hat a+\hat a^\dagger)^2\rangle-\langle\hat a+\hat a^\dagger\rangle^2
		=1$ is used in Eq.~\eqref{eq:rminus}.
		
		\emph{f. Consistency of the normalisation.}
		The norm of the excited component follows independently from the spin factor
		\begin{equation}
			\langle(\Jx)^2\rangle_0
			=\tfrac14\langle g|(\hat J_++\hat J_-)^2|g\rangle=\frac{N}{4},
			\label{eq:A6}
		\end{equation}
		giving $\|\hat V|g,\alpha\rangle\|^2
		=(4/N)\langle(\Jx)^2\rangle_0\langle(\hat a+\hat a^\dagger)^2\rangle=4\nbar+1$,
		consistent with $p=\lambda^2t^2(4\nbar+1)$ in Eq.~\eqref{eq:block}.
		
		\emph{g. Absence of odd orders.}
		In the real basis $\{|J,m\rangle\otimes|n\rangle\}$ the
		Hamiltonian~\eqref{eq:dicke} is real symmetric, and for $\alpha$ real the
		initial state has real coefficients. Hence
		\begin{equation}
			\rhoB(-t)=\rhoB(t)^{*},
			\label{eq:Ag}
		\end{equation}
		which has the same diagonal and the same spectrum as $\rhoB(t)$; both terms
		of the ergotropy are therefore even in $t$, and the first correction to
		Eq.~\eqref{eq:eps_short} is $\mathcal{O}((\lambda t)^4)$.
		
		\emph{h. Closure of the Heisenberg equations.}
		For any Hamiltonian $\hat H=\bm B(t)\cdot\bm{\hat J}$ linear in the collective
		spin, $[\hat J_i,\hat J_j]=i\epsilon_{ijk}\hat J_k$ gives
		\begin{equation}
			\frac{d\langle\hat J_i\rangle}{dt}
			=i\langle[\hat H,\hat J_i]\rangle
			=\epsilon_{ijk}B_j\langle\hat J_k\rangle
			=\left(\bm B\times\langle\bm{\hat J}\rangle\right)_i ,
			\label{eq:A7}
		\end{equation}
		a closed system in the three components $\langle\hat J_i\rangle$ with no higher
		moments. Such an $\hat H$ generates a rotation in $SU(2)$, under which spin
		coherent states are mapped to spin coherent states and
		$|\langle\bm{\hat J}\rangle|=J$ is conserved; this justifies replacing the
		operator dynamics by the Bloch equation
		$\dot{\bm r}=\bm\omega\times\bm r$ used in Lemma~\ref{prop:global}.
		
		\emph{i. Angular-velocity inequality.}
		From $\dot r_z=\omega_xr_y$ [Eq.~\eqref{eq:rzdot}], $r_z=-\cos\theta$ and
		$|r_y|\le\sin\theta$, one obtains
		$\sin\theta\,|\dot\theta|\le|\omega_x|\sin\theta\le2\GamN\sin\theta$. Dividing
		by $\sin\theta$ where it does not vanish, and extending through the isolated
		zeros $\theta\in\{0,\pi\}$ by continuity of $\dot\theta$, gives
		$|\dot\theta|\le2\GamN$ for all $t$, and hence $\theta(t)\le2\GamN t$.
		
		\emph{j. Range of validity of the envelope.}
		Equation~\eqref{eq:global_bound} holds on $0\le t\le\pi/(2\GamN)$, where
		$\sin^2(\GamN t)$ is monotone; beyond this interval it ceases to bound
		$\eps(t)$, as $\sin^2(\GamN t)$ decreases while $\eps(t)$ need not. This does
		not weaken Theorem~\ref{thm:QSL}: if the target is reached after
		$\pi/(2\GamN)$, then $\GamN\taus>\pi/2\ge\arcsin\sqrt\eps$ and the speed limit
		holds regardless of the envelope, which is the second case of the proof of
		Theorem~\ref{thm:QSL}.
		

		\section{Relation to power-based charging bounds}
		\label{app:powerbound}

		It is instructive to contrast the present speed limit with the bounds on
		instantaneous charging power that have shaped the recent
		literature~\cite{Gyhm2022}. Those bounds limit the rate at which energy
		enters the battery at each instant; ours limits the time to reach a
		prescribed \emph{ergotropy}. The two are not interchangeable, and making
		the distinction precise both clarifies what is new here and shows why a
		power bound, however tight, cannot by itself answer the timing question. In
		this appendix we derive, for the Dicke battery, the strongest timing
		statement that the power bound of Ref.~\cite{Gyhm2022} produces, and show
		that it is strictly weaker than the bound obtained in the main text, and
		weaker even than the chord estimate $\sqrt{\varepsilon}$ that is that
		bound's small-angle limit.
		
		The power bound of Ref.~\cite{Gyhm2022} replaces the naive operator-norm
		estimate $|P|\le\lVert[\hat H_B,\hat V]\rVert\le2\lVert\hat H_B\rVert\,
		\lVert\hat V\rVert$ by one controlled by the largest energy $\Delta E$ that
		the drive can move in a single step, defined as the smallest number such
		that the driving term has vanishing matrix elements between any two levels
		of the battery Hamiltonian separated by more than $\Delta E$. For a
		dynamics generated by $\hat V$ it reads
		\begin{equation}
			|P(t)|\le\tfrac12\,\Delta E\,
			\bigl\lVert\hat V-v_{\min}\mathbb{1}\bigr\rVert ,
			\label{eq:app-gyhm}
		\end{equation}
		with $v_{\min}$ the smallest eigenvalue of $\hat V$. Here the evolution is
		generated by the full $\hat H_D$ rather than by the interaction alone;
		Ref.~\cite{Gyhm2022} bounds $\lVert[\hat H_B,\hat V]\rVert$ by
		$\Delta E\inf_{\hat D}\lVert\hat V-\hat D\rVert$ over all $\hat D$ diagonal
		in the eigenbasis of $\hat H_B$, and since $\hat H_B$ and $\hat H_C$ are
		both diagonal in $\{|J,m\rangle\otimes|n\rangle\}$, the choice
		$\hat D=\hat H_B+\hat H_C+\lambda\mathbb{1}$ removes them from the
		generator and returns Eq.~\eqref{eq:app-gyhm} with $\hat V$ replaced by
		$\hat H_{\rm int}$. To evaluate it we must identify the step $\Delta E$ and
		the norm.
		The power is the rate of change of the stored energy, registered by
		$\hat H_B=\omega_0\hat J_z$,
		\begin{equation}
			P(t)=\operatorname{Tr}\!\big(\hat H_B\dot{\hat\rho}\big)
			=-i\operatorname{Tr}\!\big([\hat H_B,\hat H_D]\hat\rho\big)
			=-i\operatorname{Tr}\!\big([\hat H_B,\hat H_{\mathrm{int}}]\hat\rho\big),
			\label{eq:app-P}
		\end{equation}
		where $[\hat H_B,\hat H_B]=0$ and $[\hat H_B,\hat H_C]=0$ remove the bare
		battery and free-charger terms, leaving the interaction
		$\hat H_{\mathrm{int}}=\tfrac{2\lambda}{\sqrt N}(\hat a+\hat a^{\dagger})\hat J_x$
		as the generator that transfers energy.
		
		The step $\Delta E$ follows from the selection rule of the collective spin.
		Writing $\hat J_x=\tfrac12(\hat J_{+}+\hat J_{-})$ and using
		\begin{equation}
			\hat J_{\pm}|J,m\rangle=\sqrt{J(J{+}1)-m(m{\pm}1)}\;|J,m{\pm}1\rangle,
			\label{eq:app-ladder}
		\end{equation}
		the generator shifts the collective quantum number by one unit,
		$m\to m\pm1$, and nothing else. Since the eigenvalues of
		$\hat H_B=\omega_0\hat J_z$ are $\omega_0 m$, the levels connected by
		$\hat H_{\mathrm{int}}$ are separated by exactly one quantum, so the largest
		gap the drive bridges is
		\begin{equation}
			\Delta E=\!\!\max_{\langle E_j|\hat H_{\mathrm{int}}|E_k\rangle\neq0}\!\!
			|E_j-E_k|=\omega_0 ,
			\label{eq:app-DE}
		\end{equation}
		independent of $N$. This is the crux: the relevant scale is the
		single-quantum spacing $\omega_0$, not the full spectral width
		$\lVert\hat H_B\rVert=N\omega_0$. The collective character of the battery
		does not enlarge the step the drive can take, only the number of steps that
		must be taken to charge it.
		
		The norm is evaluated in the same classical-field limit used in the main
		text. Replacing the charger quadrature by its mean,
		$\hat a+\hat a^{\dagger}\to\langle\hat a+\hat a^{\dagger}\rangle=2\sqrt{\bar n}$,
		the generator becomes
		\begin{equation}
			\hat H_{\mathrm{int}}\longrightarrow
			\frac{2\lambda}{\sqrt N}\,2\sqrt{\bar n}\,\hat J_x=2\Gamma_N\hat J_x,
			\qquad \Gamma_N=2\lambda\sqrt{\bar n/N}.
			\label{eq:app-Hcl}
		\end{equation}
		Because $\hat J_x$ is unitarily equivalent to $\hat J_z$, its spectrum is
		$\{-J,\dots,J\}=\{-N/2,\dots,N/2\}$; hence $2\Gamma_N\hat J_x$ has extreme
		eigenvalues $\pm N\Gamma_N$, so $v_{\min}=-N\Gamma_N$ and
		\begin{equation}
			\tfrac12\bigl\lVert\hat H_{\mathrm{int}}-v_{\min}\mathbb{1}\bigr\rVert
			=\tfrac12\bigl(N\Gamma_N-(-N\Gamma_N)\bigr)=N\Gamma_N .
			\label{eq:app-norm}
		\end{equation}
		Combining the step \eqref{eq:app-DE} with the potential
		\eqref{eq:app-norm}, the maximum power is
		\begin{equation}
			P_{\max}=\Delta E\cdot
			\tfrac12\bigl\lVert\hat H_{\mathrm{int}}-v_{\min}\mathbb{1}\bigr\rVert
			=\omega_0\cdot N\Gamma_N=N\omega_0\,\Gamma_N ,
			\label{eq:app-Pmax}
		\end{equation}
		so that power and speed are governed by the same composite rate $\Gamma_N$.
		
		A power ceiling constrains how much energy can accumulate, but energy is
		not work. Integrating $P\le P_{\max}$ over the charging interval bounds the
		delivered energy,
		\begin{equation}
			E(\tau^{*})-E(0)=\int_{0}^{\tau^{*}}\!P(s)\,ds
			\le P_{\max}\,\tau^{*}=N\omega_0\,\Gamma_N\,\tau^{*}.
			\label{eq:app-int}
		\end{equation}
		The battery starts in its ground state, so $E(\tau^{*})-E(0)$ is the energy
		actually supplied, and the ergotropy, i.e. the work extractable relative to the
		passive state, cannot exceed it,
		\begin{equation}
			W(\tau^{*})\le E(\tau^{*})-E(0).
			\label{eq:app-WE}
		\end{equation}
		At the charging time the target is met, $W(\tau^{*})=\varepsilon N\omega_0$,
		so \eqref{eq:app-int} and \eqref{eq:app-WE} give
		\begin{equation}
			\varepsilon N\omega_0=W(\tau^{*})\le E(\tau^{*})-E(0)
			\le N\omega_0\,\Gamma_N\,\tau^{*},
			\label{eq:app-chain}
		\end{equation}
		and the common factor $N\omega_0$ cancels, leaving
		\begin{equation}
			X\equiv\Gamma_N\tau^{*}\ge\varepsilon .
			\label{eq:app-Xeps}
		\end{equation}
		
		Equation~\eqref{eq:app-Xeps} is the strongest timing statement the power
		bound yields; it is the \emph{power (linear-work) estimate}, treating the
		required work as if it accrued along the straight line between initial and
		final states, at the maximum rate, for the entire duration. It is weaker
		than the bound of the main text at every target, and the two are separated
		by the chord estimate $\sqrt{\varepsilon}$ that lies between them,
		\begin{equation}
			\arcsin\sqrt{\varepsilon}\;\ge\;\sqrt{\varepsilon}\;\ge\;\varepsilon,
			\qquad \varepsilon\in(0,1].
			\label{eq:app-ineq}
		\end{equation}
		The three quantities have distinct meanings on the Bloch sphere:
		$\arcsin\sqrt{\varepsilon}$ is the arc actually traversed and sets the exact
		bound of the main text; $\sqrt{\varepsilon}=\sin(\theta/2)$ is the chord
		joining the initial and final states, and is also the small-angle limit of
		the arc, $\arcsin\sqrt{\varepsilon}\to\sqrt{\varepsilon}$ as
		$\varepsilon\to0$, the two coinciding to leading order; and the power
		estimate $\varepsilon$ lies strictly below the chord. The first inequality
		follows from $\arcsin x\ge x$ on $[0,1]$ (with equality only at $x=0$), the
		second from $\sqrt{\varepsilon}\,(1-\sqrt{\varepsilon})\ge0$ on $(0,1]$. The
		power bound therefore does not recover the chord: it falls short of
		$\sqrt{\varepsilon}$ by a factor $\varepsilon^{-1/2}$, and short of the arc
		by the same factor to leading order, a gap that widens without bound as
		$\varepsilon\to0$ while the arc-to-power ratio
		$\arcsin\sqrt{\varepsilon}/\varepsilon$ remains strictly above unity all the
		way to full charging. Within the classical-field comparison made here,
		this is the sense in which the present result is not a restatement of a
		power bound: the improvement over it grows as $\varepsilon^{-1/2}$ rather than being a constant factor. Figure~\ref{fig:comparison} shows the arc
		and power bounds together with the numerically extracted charging
		times, whose lower envelope tracks the arc bound $\arcsin\sqrt{\varepsilon}$
		and not the weaker power estimate $\varepsilon$.
		\begin{figure}[t]
			\centering
			\includegraphics[width=\columnwidth]{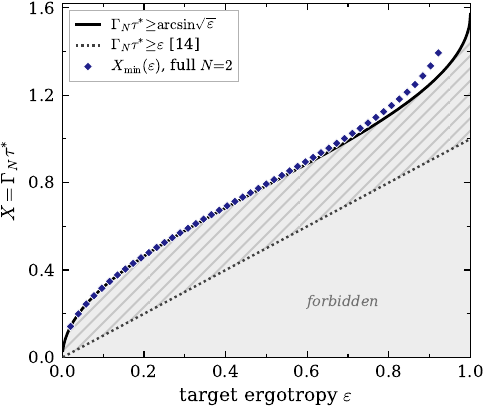}
			\caption{%
				Ergotropy speed limit compared with the instantaneous-power bound of
				Ref.~\cite{Gyhm2022}, both in the classical-field limit with
				$\Gamma_N=2\lambda\sqrt{\bar n/N}$, in the collapse variable
				$X=\Gamma_N\tau^{*}$. The power bound integrates to $X\ge\varepsilon$
				(dotted); the present bound $X\ge\arcsin\sqrt{\varepsilon}$ (solid) is
				tighter at every target ergotropy (hatched region), the improvement
				factor $\arcsin\sqrt{\varepsilon}/\varepsilon$ growing as
				$\varepsilon^{-1/2}$ as $\varepsilon\to0$. Diamonds: numerical lower
				envelope $X_{\min}(\varepsilon)$ from the full $N=2$ Dicke dynamics
				($N_{\mathrm{Fock}}=60$, charging times by root-finding). The
				shaded region is dynamically inaccessible.%
			}
			\label{fig:comparison}
		\end{figure}
		The reason a power bound must fall short is worth stating plainly, because
		it is the same reason the result is not merely a restatement of
		semiclassical rotation. Two distinct approximations separate the power
		estimate from the truth. The first is dynamical: peak power is assumed in
		\eqref{eq:app-int} to persist for the whole interval, whereas the
		collective rotation reaches its maximum rate only momentarily, at a quarter
		of the Rabi period, and is slower on either side. The second is
		thermodynamic and more fundamental: it is the single inequality
		\eqref{eq:app-WE}. A power bound tracks stored energy, which for a coherent
		drive scales as the second quadrature moment
		$\langle(\hat a+\hat a^{\dagger})^{2}\rangle=4\bar n+1$; the ergotropy
		tracks only the extractable part, the squared first moment
		$\langle\hat a+\hat a^{\dagger}\rangle^{2}=4\bar n$. The difference is the
		zero-point contribution of the charger, deposited as energy but locked into
		the passive state and yielding no work. It is exactly this quantity that a
		bound on energy transfer counts and an ergotropy bound must not, and it is
		the reason the rotation angle, a kinematic object indifferent to how the
		energy is distributed, can be converted into a genuine bound on the
		charging time only after the passive-state subtraction has been made.
		
		A final remark places the comparison on an honest footing. The evaluation
		\eqref{eq:app-norm} presumes a bounded drive: the amplitude entering
		$\Delta E$ and the norm is the classical field $2\sqrt{\bar n}$ of
		\eqref{eq:app-Hcl}. For the full quantised drive the quadrature
		$\hat a+\hat a^{\dagger}$ is unbounded on Fock space,
		$\lVert\hat H_{\mathrm{int}}\rVert$ diverges, and the power bound in its
		original form does not constrain the dynamics at all. The same restriction
		applies to the speed limit of the main text, whose angular-velocity argument
		likewise closes only once the quadrature is replaced by its bounded classical
		value. What differs is what the two bounds deliver once that replacement is
		made: Eq.~\eqref{eq:app-Xeps} constrains the accumulated energy and is
		recovered as the weakest member of the ordering~\eqref{eq:app-ineq}, whereas
		Eq.~\eqref{eq:global_bound} constrains the swept angle, and it is the
		passive-state subtraction that converts the latter into a statement about
		extractable work.

	\end{document}